\documentclass[11pt,a4paper,UKenglish,numberwithinsect]{article}
 
\usepackage[left=2.5cm,right=2.5cm,top=2.5cm,bottom=3cm]{geometry}
\usepackage{tikz}
\usetikzlibrary{quotes,graphs}
\usepackage{amsmath,amsthm,amssymb}
\usepackage[pdfborder={0 0 0}]{hyperref}

\title{Computing Hitting Set Kernels By AC$^0$-Circuits}

\author{Max Bannach \and Till Tantau}
\date{%
  Institute for Theoretical Computer Science,\\
  Universit\"at zu L\"ubeck\\
  L\"ubeck, Germany \\
  \texttt{\{bannach,tantau\}@tcs.uni-luebeck.de}
}

\theoremstyle{theorem}
\newtheorem{theorem}{Theorem}[section]
\newtheorem{corollary}[theorem]{Corollary}
\newtheorem{lemma}[theorem]{Lemma}

\theoremstyle{plain}
\newtheorem{example}[theorem]{Example}
\newtheorem{fact}[theorem]{Fact}
\newtheorem{problem}[theorem]{Problem}
\newtheorem{definition}[theorem]{Definition}

\newcommand\Class[1]{%
  \mathchoice%
  {\text{\normalfont\fontsize{9pt}{10pt}\selectfont$\mathrm{#1}$}}%
  {\text{\normalfont\fontsize{9pt}{10pt}\selectfont$\mathrm{#1}$}}%
  {\text{\normalfont$\mathrm{#1}$}}%
  {\text{\normalfont$\mathrm{#1}$}}%
}

\newcommand\Para{\mathrm{para\text-}}

\newcommand{\Lang}[1]{\text{\normalfont\textsc{#1}}}
\newcommand\PLang[1][]{p_{{#1}}\text-\penalty15\Lang}

\tikzset{
  hypergraph/.style={
    semithick
  },
  vertex/.pic={
    \fill [black!25, rounded corners=1mm] (0,0) rectangle (#1,1.5em);
    \node [anchor=mid west,font=\footnotesize] at (0,0.75em) {\tikzpictext};
    \coordinate (-1) at (4mm,0.75em);
    \coordinate (-1') at (4mm+4mm,0.75em);
    \coordinate (-2) at (4mm+6mm,0.75em);
    \coordinate (-3') at (4mm+12mm-4mm,0.75em);

    \coordinate (-4') at (2cm+4mm+4mm,0.75em);
    \coordinate (-5) at (2cm+4mm+6mm,0.75em);
    \coordinate (-6') at (2cm+4mm+12mm-4mm,0.75em);

    \coordinate (-7') at (4cm+4mm+4mm,0.75em);
    \coordinate (-8) at (4cm+4mm+6mm,0.75em);
    \coordinate (-9') at (4cm+4mm+12mm-4mm,0.75em);

    \coordinate (-bottom) at (2.5mm,0.8mm);
    \coordinate (-top) at (2.5mm,1.5em-0.8mm);
  },
  hyperedge/.style={very thick,cap=round},
  dot/.style={outer sep=0pt,fill,inner sep=0pt,minimum size=1mm,circle}
}

\begin{document}
\maketitle

\begin{abstract}
  Given a hypergraph $H = (V,E)$, what is the smallest subset $X
  \subseteq V$ such that $e \cap X \neq \emptyset$ holds for all $e \in E$?
  This problem, known as the \emph{hitting set problem,} is a
  basic problem in parameterized complexity theory. There are well-known
  kernelization algorithms for it, which get a hypergraph~$H$ and a
  number~$k$ as input and output a hypergraph~$H'$ such that (1)
  $H$ has a hitting set of size~$k$ if, and only if, $H'$ has such a
  hitting set and (2) the size of $H'$ depends only on $k$
  and on the maximum cardinality $d$ of edges in~$H$. The
  algorithms run in polynomial time, but are highly
  sequential. Recently, it has been shown that one of them can be parallelized
  to a certain degree: one can compute hitting set kernels in parallel
  time $O(d)$ -- but it was conjectured that this is
  the best parallel algorithm possible. We
  refute this conjecture and show how hitting set kernels can be
  computed in \emph{constant} parallel time. For our proof, we
  introduce a new, generalized notion of hypergraph sunflowers and
  show how iterated applications of the color coding technique can
  sometimes be collapsed into a single application.
\end{abstract}

\section{Introduction}

The hitting set problem is the following combinatorial problem: Given
a hypergraph $H = (V,E)$ as input, consisting of a set $V$ of vertices
and a set $E$ of \emph{hyperedges} with $e \subseteq V$ for all $e\in
E$, find a set $X\subseteq V$ of minimum size that ``hits'' all
hyperedges $e\in E$, that is, $e \cap X \neq \emptyset$. Many problems
reduce to the hitting set problem, including 
the vertex cover problem (it is exactly the special case where all
edges have size $|e| = 2$) and the dominating set problem (a dominating
set of a graph is exactly a hitting set of the hypergraph whose hyperedges
are the closed neighborhoods of the graph's vertices). The
computational complexity of the hitting set problem is thus of 
interest both in classical complexity theory and in parameterized complexity theory. 

The first result on the parameterized complexity of
the hitting set problem was an efficient \emph{kernelization algorithm}
for this problem restricted to edges of cardinality
three~\cite{NiedermeierR03}. This was later improved to a
kernelization for the $d$-uniform version (all hyperedges have size
exactly~$d$)~\cite{FlumG06}, which is based on the so-called Sunflower 
Lemma~\cite{ErdosR60}. We will later have a closer look at this
algorithm; at this point let us just summarize its main 
idea by ``repeatedly find sunflowers and replace them by their
cores until there are no more sunflowers.'' The Sunflower Lemma tells
us that this algorithm will stop only when the input graph has been
reduced to a kernel. 
The just-sketched kernelization algorithm is highly sequential, but Chen et
al.~\cite{ChenFH2017} have recently shown that it can be parallelized:
Instead of reducing sunflowers 
one-at-a-time, one can replace all sunflowers in a hypergraph by their
cores simultaneously in constant parallel time. This process only
needs to be repeated $d(H) = \max_{e\in E} |e|$ times, leading to a parallel algorithm running in
time $O(d(H))$. However, there were good reasons to believe that
this algorithm is essentially the best possible (we will later
discuss them) and Chen et al.\ conjectured that the
hitting set problem does not admit a kernelization algorithm running
in constant parallel time (that is, in time completely independent of
the input graph).

\paragraph*{Our Contributions.}

In the present paper we refute the conjecture of Chen et al.\ and show
that there is a constant parallel time kernelization algorithm for the
hitting set problem:

\begin{problem}{$\PLang[k,d]{hitting-set}$}\label{kd-hs}
  \begin{description}\parskip0pt\itemsep0pt
    \item[Instance:] A hypergraph $H=(V,E)$ and a number $k\in\mathbb{N}$.
    \item[Parameter:] $k + d(H)$
    \item[Question:] Does $H$ have a hitting set $X$ with $|X| \le k$?
  \end{description}
\end{problem}

\begin{theorem}[Main Theorem]\label{theorem-main}
  There is a \textsc{dlogtime}-uniform $\Class{AC}^0$-circuit family
  that maps every hypergraph $H = (V,E)$ and number $k$ to a new
  hypergraph $H' = (V,E')$ that has the same size-$k$ hitting sets as
  $H$, has $d(H') \le d(H)$, and has $|E'| \le f(k,d(H))$ for some
  fixed computable function~$f$.
\end{theorem}
Let us stress at this point that the $\Class{AC}^0$-family from the
theorem really has a size that is polynomial in the input length (no
exponential or even worse dependency on the parameters) and has a
depth that is completely independent of the input.  
The hypergraph $H'$ has the same vertex set~$V$ as~$H$ -- a
feature shared by all hypergraphs considered in this paper that
simplifies the presentation. However, since $V$ is still ``large,''
the circuit is not quite a kernelization algorithm. Fortunately, this is
easy to fix by replacing the vertex set of~$H'$ by $V'=\bigcup_{e\in
  E'} e$, yielding the following corollary: 
\begin{corollary}[Constant-Time Kernelization]\label{corollary-main}
  There is a \textsc{dlogtime}-uniform $\Class{AC}^0$-circuit family
  that computes a kernel for every instance for
  $\PLang[k,d]{hitting-set}$.   
\end{corollary}
The theorem and corollary imply that all problems that can be reduced to
$\PLang[k,d]{hitting-set}$ via a parameter-preserving
$\Class{AC}^0$-reduction admit a kernelization computable by an
$\Class{AC}^0$-circuit family. This includes  $\PLang[k]{vertex-cover}$, which
is just $\PLang[k,d]{hitting-set}$ with $d$ fixed at~$2$;
$\PLang[k]{triangle-removal}$, where the objective is to remove at
most $k$ vertices from an undirected graph so that no triangles
remain; and also $\PLang[k,\mathrm{deg}]{dominating-set}$, where we
must find a dominating set of size at most~$k$ in an undirected graph
and we parametrize by $k$ and the maximum degree of the vertices. 

Our proof of the main theorem requires the development of two new
ideas, which we believe may also be useful in other situations.
The above-mentioned parallel kernelization algorithm for the hitting
set problem with runtime $O(d(H))$ essentially does the following: ``Repeat $d(H)$ times:
replace all sunflowers of size $k+1$ by their cores'' and the
difficult task in each of the $d(H)$ iterations is to find the
sunflowers. It turns out that this can be done in constant parallel
time using the \emph{color coding} technique~\cite{AlonYZ95} and it has been shown
in~\cite{BannachST15} and again in~\cite{ChenFH2017} that
this technique can be implemented in constant time. Our first idea for
turning the circuits depth from $O(d)$ into $O(1)$ is to
\emph{collapse the color codings from the $d$ rounds into a single
  application of the color coding technique:} 
Instead of applying color coding in each round to filter and describe
``objects,'' we would like to apply one global application of color
coding that already contains the internal colorings and does away
with the intermediate objects.

Unfortunately, there does not appear to be a simple (or any) way of
actually collapsing the colorings used when we ``replace
all sunflowers by their cores'': The coloring coding technique is good
at imposing requirements of the form ``these objects must be
disjoint,'' but cannot impose requirements of the form ``these objects
must be the same.'' For this reason, as our second new idea, we
develop a generalization of the notion of a sunflower (which we dub
``pseudo-sunflowers'') that is tailored to the collapsing of color
coding.  

\paragraph*{Related Work.}

The sequential kernelization algorithm for the hitting set problem
based on the Sunflower Lemma  has been known for a longer
time~\cite{FlumG06}, but there have been recent improvements
that bring down the runtim to linear time~\cite{vanBevern2014}. A
parallel version has recently been studied by Chen 
et al.~\cite{ChenFH2017} and they show how kernels for
$\PLang[k,d]{hitting-set}$ can be computed by circuits of
depth~$O(d(H))$. Chen et al.\ also conjecture that the circuit depth of
$O(d(H))$ is unavoidable (which we refute).

The results of this paper fit into the larger, fledgling field of
parallel parameterized complexity theory, which has already been
studied both 
from a practical~\cite{AbuKhzamLSS2006} and a theoretical point of
view~\cite{Cesati:1998fe}. First results go back to
research on \emph{parameterized logarithmic space}~\cite{CaiCDF97,
  ChenFG2003, FlumG2002}, since it is known from classical complexity
theory 
that problems that are solvable with such a resource bound can also be parallelized.
A more structured analysis of parameterized space and circuit
classes was later made by Elberfeld et.~al~\cite{ElberfeldST2014},
which addresses parallelization more directly.
Current research on parameterized parallelization -- including this
paper -- focuses on constant-time computations, that is, on a
parameterized analogue of $\Class{AC^0}$~\cite{ChenF16, ChenFH2017,
  BannachST15, BannachT16}. We remark that many previous results
(including several of the authors) boil down to showing that instead of
using a known reduction rule many 
times sequentially, one can simply apply it in parallel
``everywhere,'' but ``only once.'' In contrast, the kernelization
algorithm developed in the present paper had no previous counterpart
in the sequential setting. 

\paragraph*{Organization of This Paper.}

After a short section on preliminaries, in Section~\ref{section-known}
we review known kernelization algorithms for the hitting set problem
-- both the sequential ones and the parallel one. In
Section~\ref{section-pseudo} we discuss the 
obstacles that must be surmounted to turn the known parallel algorithm
into one that needs only constant time. Towards this aim, we introduce
the notions of pseudo-cores and pseudo-sunflowers as replacements for the
cores and sunflowers used in the known algorithms. In
Section~\ref{section-algorithm} we then argue that these
pseudo-sunflowers can be computed in constant time by
``collapsing'' multiple rounds of color coding into a single
round. Full proofs can be found in the full version of the
paper~\cite{BannachT2018}.

\section{Preliminaries}

A \emph{hypergraph} is a pair $H = (V,E)$ such that for all \emph{hyperedges
  $e \in E$} we have $e \subseteq V$. We write $V(H) = V$ and $E(H) =
E$ for the vertex and hyperedge sets of~$H$. Let $d(H) = \max_{e\in E}
|e|$. Throughout this paper, all
hypergraphs will always have the same vertex set~$V$, which is the
input vertex set. For this reason, in slight abuse of notation, for
two hypergraphs $H_1 = (V,E_1)$ and $H_2 = (V,E_2)$ we also write $H_1
\subseteq H_2$ for $E(H_1) \subseteq E(H_2)$ and $H_1 \cup H_2$ for
$(V, E(H_1) \cup E(H_2))$.

Concerning circuit classes and parallel computations, we will only
need the notion of $\Class{AC}$-circuit families, which are sequences
$C=(C_0,C_1,C_2,\dots)$ of Boolean circuits where each $C_i$ is a
directed acyclic graph whose vertices are gates such that there are $i$
input gates, the inner gates are $\land$-gates or $\lor$-gates with unbounded fan-in,
or $\lnot$-gates; and the number
of output gates is either $1$ (for decision problems) or depends on
the number of input gates (for circuits computing a function). The
\emph{size function~$S$} maps circuits to their size (number of gates)
and the \emph{depth function~$D$} maps them to their depth (longest
path from input gates to output gates). When $D(C_n) \in O(1)$ and
$S(C_n) \in n^{O(1)}$ hold, we call $C$ an 
$\Class{AC}^0$-circuit family. Concerning circuit uniformity, all circuit
families in this paper will be \textsc{dlogtime} uniform, which is 
the strongest notion of uniformity commonly
considered~\cite{BarringtonIS88}: there is a \textsc{dtm} that on
input of $\operatorname{bin}(i)\#\operatorname{bin}(n)$, 
where $\operatorname{bin}(x)$ is the binary encoding of~$x$, outputs
the $i$th bit of a suitable encoding of $C_{n}$ in at most $O(\log n)$
steps.

Even though this paper is about a parallel kernelization algorithm, we
will need only little from the machinery of parallel parameterized
complexity theory. We do need the following notions: A
\emph{parameterized problem} is a pair $(Q,\kappa)$ where $Q \subseteq
\Sigma^*$ is a language and $\kappa$ is a function $\kappa \colon
\Sigma^* \to \mathbb N$ that is computable by a
\textsc{dlogtime}-uniform $\Class{AC}^0$-circuit family. When we write
down a parameterized problem such as 
$\PLang[k,d]{hitting-set}$, the indices of ``$p$'' (for
``parameterized'') indicate which 
parameter function $\kappa$ we mean. A \emph{kernelization}
for a parameterized problem $(Q,\kappa)$ is a function~$K$ that maps
every instance $x \in \Sigma^*$ to a new instance $K(x) \in \Sigma^*$
such that for all $x\in \Sigma^*$ we have  (1)~$x \in Q \iff K(x) \in Q$ and
(2)~$|K(x)| \le f(\kappa(x))$ for some fixed computable function~$f$.

A parameterized problem $(Q,\kappa)$ lies in
$\Class{FPT}$ if $x\in Q$ can be decided by a sequential algorithm running
in time $f(\kappa(x))\cdot|x|^{O(1)}$ for a computable
function~$f$. The $\Class{AC}^0$-analogue of $\Class{FPT}$ is
the class $\Para\Class{AC}^0$. It contains all problems
$(Q,\kappa)$ for which there is a circuit family
$(C_{n,k})_{n,k\in\mathbb N}$ such that for all inputs $x$ we have
$C_{|x|,\kappa(x)}(x) = 1$ if, and only if, $x\in 
Q$, and  $D(C_{n,k}) \in O(1)$ and $S(C_{n,k}) \in f(k)\cdot
n^{O(1)}$. It is well-known that $(Q,\kappa) \in \Class{FPT}$ holds
if, and only if, $Q$ is decidable and there is a kernelization for $(Q,\kappa)$ that is 
computable in polynomial time. The same proof as for the
polynomial-time case also shows that we have $(Q,\kappa)
\in \Para\Class{AC}^0$ if, and only if, $Q$ is decidable and $(Q,\kappa)$ has a
kernelization that can be computed by an $\Class{AC}^0$-circuit
family. (We stress once more that this means that the kernelization is
a normal $\Class{AC}^0$-circuit family, having size $S(C_n) \in n^{O(1)}$.)

We will use the \emph{color coding technique} a lot. First introduced
in~\cite{AlonYZ95}, it has recently been shown to work in the context
of constant time computations~\cite{BannachST15,ChenFH2017}. The key
observation underlying this technique is the following: Suppose we are
given a set of $n$ elements and suppose you have $k$ special elements
$x_1,\dots,x_k$ together with some specific colors $c_1,\dots,c_k$ for
them ``in mind''. Then we can compute a set~$\Lambda$ of ``candidate
colorings'' of all elements of the set such that at least one
$\lambda\in \Lambda$ colors each ``in mind'' vertex~$x_i$ with the
``desired'' color $c_i$, that is $\lambda(x_i) = c_i$. Formally, the
following holds (the original  version of this lemma due to Alon
et.\,al~\cite{AlonYZ95} is equivalent to the statement below -- only
without any depth guarantees):  

\begin{fact}[Color Coding Lemma, \cite{BannachST15}]\label{fact-cc}
  There is a \textsc{dlogtime}-uniform family
  $(C_{n,k,c})_{n,k,c\in\mathbb N}$ of $\Class{AC}$-circuits without
  inputs such that each $C_{n,k,c}$  
  \begin{enumerate}
  \item outputs a set $\Lambda$ of functions $\lambda \colon
    \{1,\dots,n\} \to \{1,\dots,c\}$ (coded as a
    sequence of function tables) with the property that for any $k$
    mutually distinct $x_1,\dots,x_k \in \{1,\dots,n\}$ and any $c_1,\dots,c_k
    \in \{1,\dots,c\}$ there is a function $\lambda\in \Lambda$ with $\lambda(x_i) =
    c_i$ for all $i \in \{1,\dots,k\}$, 
  \item has constant depth (independent of $n$, $k$, or $c$), and
  \item has size at most $O(\log c \cdot c^{k^2} \cdot k^4 \cdot n\log^2 n)$.
  \end{enumerate}
\end{fact}



\section{Known Kernelization Algorithms for the Hitting Set Problem}
\label{section-known}

\subsection{Known Sequential Kernelization Algorithms}

The knwon sequential kernelization algorithms for
$\PLang[k,d]{hitting-set}$ are based on the so-called
\emph{Sunflower Lemma.} The perhaps simplest application of this lemma
is to repeatedly collapses sufficiently large sunflowers to their
cores until there are no longer any large sunflowers in the graph and,
then, the Sunflower Lemma tells us that the graph ``cannot be 
very large.'' In detail, the definitions and algorithm are as
follows:

\begin{definition}[Sunflower]
  A \emph{sunflower $S$ with core $C$} is a set of proper supersets
  of~$C$ such that for 
  any two distinct $p,q \in S$  we have $p \cap q = C$. The elements
  of a sunflower are called \emph{petals.} A \emph{sunflower in a
    hypergraph} is a sunflower whose petals are hyperedges of the
  hypergraph. 
\end{definition}


\begin{fact}[Sunflower Lemma~\cite{ErdosR60}]\label{fact-sunflower}
  Every hypergraph $H$ with more than $k^{d(H)}\cdot
  d(H)!$ hyperedges contains a sunflower of size~$k+1$.
\end{fact}
The importance of the Sunflower Lemma for the hitting set problem lies
in the following observation: Suppose a hypergraph $H$ contains a
sunflower $S$ of size at least $k+1$. Then $H$ has a size-$k$ hitting
set if, and only if, the hypergraph obtained from $H$ by removing all
petals of the sunflower and adding its core has such a hitting set (we
cannot hit the $k+1$ petals in the sunflower using only $k$ vertices
without using at least one vertex of the core; thus, we hit all petals
if, and only if, we hit the core). In other words, replacing a
sunflower of size $k+1$ by its core is a reduction rule for the
hitting set problem; and if we can no longer apply this rule, the
Sunflower Lemma tells us that the hypergraph's size is bounded by a
function that depends only on $k$ and $d(H)$~--~in other words, it is
a kernel.

The just-described kernelization algorithm is simple, but ``very
sequential.'' It is, however, not too difficult to turn it into a
more parallel algorithm -- at least, as long as $d(H)$ is fixed. This
was first noted by Chen et al.~\cite{ChenFH2017} and we explain
the ideas behind their proof below, rephrased for the
purposes of the present paper.

A better sequential kernelization algorithm has
recently~\cite{vanBevern2014} been proposed (it runs in time
$O(2^{d(H)} |E|)$, which is linear from a parameterized point of view)
-- but the algorithm is arguably ``even more sequential''
and does not lend itself to easy parallelization.

\subsection{Known Parallel Kernelization Algorithm}

The first step towards a parallel kernelization is the
observation that we can compute many cores in parallel. Given a
hypergraph $H = (V,E)$ and a number~$k$, let a \emph{$k$-core in~$H$} be a
core~$C$ of a sunflower in~$H$ with more than
$k$~petals. Let
$\operatorname{\mathit k-cores}(H) = \bigl(V, \{C \mid \text{$C$ is a
  $k$-core in $H$}\}\bigr)$. While in the sequential algorithm we
always replace one sunflower by its core, we now replace \emph{all}
sunflowers by their cores. This leaves behind some hyperedges, but the
Sunflower Lemma will show that their number is ``small.''
Unfortunately, the set of cores itself may still be
large and we need to apply the replace-all-sunflowers-by-cores
operation repeatedly.  This process \emph{does} stop after at most $d(H)$
rounds since the \emph{size} of the cores decreases by $1$ in each
round and, hence, after $d(H)$ rounds it has shrunk to~$0$.

\def\vertices{
  \scoped[yscale=.8]{
    \pic (a) ["$a$"] at (0,0)      {vertex = 5.7cm};
    \pic (b) ["$b$"] at (0,0.75)   {vertex = 5.7cm};
  }
  \verticesup
}
\def\verticesup{
  \scoped[yscale=.8]{
    \pic (c) ["$c$"] at (0,1.9)      {vertex = 1.7cm};
    \pic (d) ["$d$"] at (2,1.9)      {vertex = 1.7cm};
    \pic (e) ["$e$"] at (4,1.9)      {vertex = 1.7cm};

    \pic (f) ["$f$"] at (0,3)    {vertex = 5mm};
    \pic (g) ["$g$"] at (0.6,3)  {vertex = 5mm};
    \pic (h) ["$h$"] at (1.2,3)  {vertex = 5mm};

    \pic (i) ["$i$"] at (2+0,3)    {vertex = 5mm};
    \pic (j) ["$j$"] at (2+0.6,3) {vertex = 5mm};
    \pic (k) ["$k$"] at (2+1.2,3)  {vertex = 5mm};

    \pic (l) ["$l$"] at (4+0,3)    {vertex = 5mm};
    \pic (m) ["$m$"] at (4+0.6,3) {vertex = 5mm};
    \pic (n) ["$n$"] at (4+1.2,3)  {vertex = 5mm};

    \pic (o) ["$o$"] at (2+0,4)    {vertex = 5mm};
    \pic (p) ["$p$"] at (2+0.6,4) {vertex = 5mm};
    \pic (q) ["$q$"] at (2+1.2,4)  {vertex = 5mm};

    \pic (r) ["$r$"] at (2+0,4.75)    {vertex = 5mm};
    \pic (s) ["$s$"] at (2+0.6,4.75) {vertex = 5mm};
    \pic (t) ["$t$"] at (2+1.2,4.75)  {vertex = 5mm};

    \pic (u) ["$u$"] at (2+0,5.5)    {vertex = 5mm};
    \pic (v) ["$v$"] at (2+0.6,5.5) {vertex = 5mm};
    \pic (w) ["$w$"] at (2+1.2,5.5)  {vertex = 5mm};
    }
  }

  \def\hyperedgesh{
    \draw [hyperedge] (u-top) -- (v-top) -- (w-top);
    
    \draw [hyperedge,r] (a-1') -- (b-1') -- (c-1') to[out=90,in=-90] (f-1) to [out=90,in=180] (u-bottom) -- (v-bottom) -- (w-bottom);
    \draw [hyperedge,r] (a-2) -- (b-2) -- (c-2) to[out=90,in=-90] (g-1) to [out=90,in=180] (r-bottom) -- (s-bottom) -- (t-bottom)
    to[out=0,in=90] ([yshift=1mm]m-1);
    \draw [hyperedge,r] (a-3') -- (b-3') -- (c-3') to[out=90,in=-90] (h-1) to [out=90,in=180] (o-bottom) -- (p-bottom) -- (q-bottom)
    to[out=0,in=90] ([xshift=-1mm]l-1) -- ([xshift=-1mm]e-1);

    \draw [hyperedge,b] (a-4') -- (b-4') -- (d-1') to[out=90,in=-90] (i-1) -- (o-1) -- (r-1) -- (u-1);
    \draw [hyperedge,b] (a-5) -- (b-5) -- (d-2) to[out=90,in=-90] (j-1) -- (p-1) -- (s-1) -- (v-1);
    \draw [hyperedge,b] (a-6') -- (b-6') -- (d-3') to[out=90,in=-90] (k-1) -- (q-1) -- (t-1) -- (w-1);

    \draw [hyperedge,g] (a-7') -- (b-7') -- (e-1') to[out=90,in=-90] (l-1);
    \draw [hyperedge,g] (a-8) -- (b-8) -- (e-2) to[out=90,in=-90] (m-1);
    \draw [hyperedge,g] (a-9') -- (b-9') -- (e-3') to[out=90,in=-90] (n-1);
  }

  \def\hyperedgeshup{
    \draw [hyperedge,r] (c-1') to[out=90,in=-90] (f-1) to [out=90,in=180] (u-bottom) -- (v-bottom) -- (w-bottom);
    \draw [hyperedge,r] (c-2) to[out=90,in=-90] (g-1) to [out=90,in=180] (r-bottom) -- (s-bottom) -- (t-bottom)
    to[out=0,in=90] ([yshift=1mm]m-1);
    \draw [hyperedge,r] (c-3') to[out=90,in=-90] (h-1) to [out=90,in=180] (o-bottom) -- (p-bottom) -- (q-bottom)
    to[out=0,in=90] ([xshift=-1mm]l-1) -- ([xshift=-1mm]e-1);

    \draw [hyperedge,b] (d-1') to[out=90,in=-90] (i-1) -- (o-1) -- (r-1) -- (u-1);
    \draw [hyperedge,b] (d-2) to[out=90,in=-90] (j-1) -- (p-1) -- (s-1) -- (v-1);
    \draw [hyperedge,b] (d-3') to[out=90,in=-90] (k-1) -- (q-1) -- (t-1) -- (w-1);

    \draw [hyperedge,g] (e-1') to[out=90,in=-90] (l-1);
    \draw [hyperedge,g] (e-2) to[out=90,in=-90] (m-1);
    \draw [hyperedge,g] (e-3') to[out=90,in=-90] (n-1);
  }

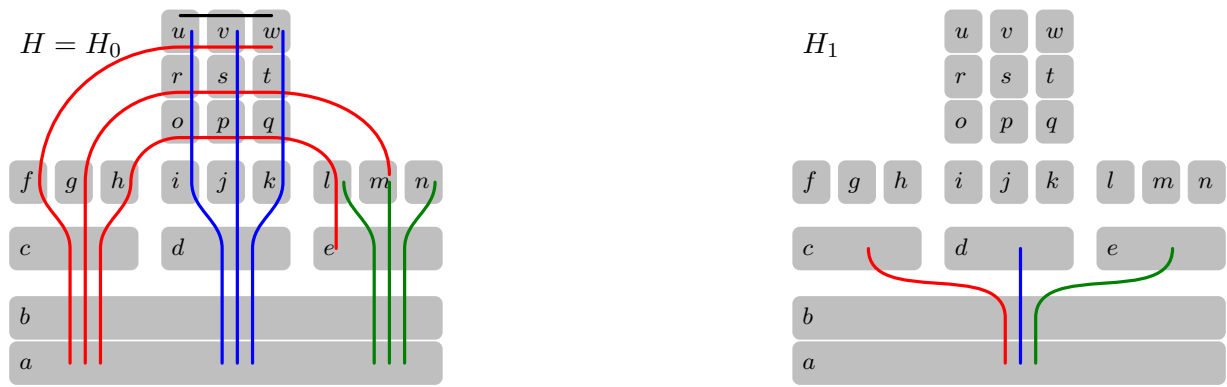
\begin{figure}[htpb]
  
  \begin{tikzpicture}

    \node [below right] at (0mm,4.8cm) {$H=H_0$};
    
    \vertices
    \scoped[r/.style=red,g/.style=green!50!black,b/.style=blue]{
      \hyperedgesh
    }

  \end{tikzpicture}
  \hfill  
  \begin{tikzpicture}

    \node [below right] at (0mm,4.8cm) {$H_1$};
    
    \vertices
    
    \draw [hyperedge,red] (a-4') -- (b-4') to[out=90,in=-90] (c-2);
    \draw [hyperedge,blue] (a-5) -- (b-5) to[out=90,in=-90] (d-2);
    \draw [hyperedge,green!50!black] (a-6') -- (b-6') to[out=90,in=-90] (e-2);
  \end{tikzpicture}
  
  \caption{
    Visualization of a hypergraph $H_0$ and of its 2-cores $H_1 =
    \operatorname{2-cores}(H_0)$. Vertices are drawn as 
    rectangles, while the ten hyperedges of~$H_0$
    are drawn as lines: they contain all vertices that they touch. For instance,
    the leftmost line starting in the vertex~$a$ in $H_0$ visualizes
    the hyperedge $\{a,b,c,f,u,v,w\}$ and the rightmost line
    visualizes the hyperedge $\{a,b,e,n\}$. The hypergraph $H_0$
    contains three sunflowers of size~$3$, visualized by the red,
    blue, and green lines, respectively. Their cores are the
    hyperedges shown in $H_1$. These cores, in turn, form a sunflower
    in~$H_1$ with core $\{a,b\}$, but note that $\{a,b\}$ is
    \emph{not} a 2-core of~$H_0$. It is the only hyperedge of~$H_2$.
  }
  \label{fig:example1}
\end{figure}

Let us now formalize these ideas a bit: Let $H_0 = H$ and let $H_{i+1}
= \operatorname{\mathit k-cores}(H_i)$. Then $H_0$ is the original
hypergraph; $H_1$ is the set of its $k$-cores; $H_2$ is the set of
$H_1$'s $k$-cores and thus the set of ``cores of cores'' of~$H$; next
$H_3$ is the set of ``cores of cores of cores'' of~$H$; and so
on, see Figure~\ref{fig:example1} for an example. In a sense, each $H_i$ is nested into the previous hypergraph,
leading to a whole sequence resembling a matryoshka doll. Below, we
define a \emph{matryoshka sequence} as a sequence that has this ``nested
in some sense'' property and then show in
Lemma~\ref{lemma-core-properties} that $(H_0,H_1,\dots)$ is, indeed,
such a matryoshka sequence: 

\begin{definition}[Matryoshka Sequence]\label{definition-matryoshka}
  A \emph{matryoshka sequence for a hypergraph $H = (V,E)$ and a
    number~$k$} is a sequence $(M_0, M_1, \dots, M_{d(H)})$ of
  hypergraphs, all of which have the same vertex set~$V$, with the
  following properties for all $i\in\{0,\dots,d(H)\}$:
  \begin{enumerate}
  \item $M_0 = H$,
  \item $d(M_i) \le d(H) - i$,
  \item $\operatorname{\mathit k-cores}(M_i) \subseteq M_{i+1}$, and
  \item every size-$k$ hitting set of~$H$ is also a hitting set
    of~$M_i$. 
  \end{enumerate}
\end{definition}

\begin{lemma}[Cores of Cores Form a Matryoshka Sequence]\label{lemma-core-properties}
  For every hypergraph~$H$ and number~$k$, the sequence
  $(H_0,\dots,H_{d(H)})$ is a matryoshka sequence for~$H$ and~$k$.  
\end{lemma}
\begin{proof}
  The first property of a matryoshka sequence is true by
  definition. The second property follows since each time we apply the operator
  $\operatorname{\mathit  
    k-cores}$ to a hypergraph, the maximum size of the hyperedges
  decreases by at least~$1$ (cores are smaller than the largest petals
  of their sunflowers). Since we start at $d(H)$, we get $d(H_i) \le
  d(H)-i$. For the third property we actually even have equality here
  by definition. The last property is proved by induction on~$i$. The
  case $i=0$ is trivial, so consider a hitting set $X$ of size~$k$ for
  $H_{i-1}$ and consider a $k$-core $C$ of $H_{i-1}$. By definition
  there must be a sunflower $\{e_1,\dots,e_{k+1}\} \subseteq
  E(H_{i-1})$ with core~$C$. If $X$ did not hit $C$ (that is, if $X
  \cap C = \emptyset$), then the size-$k$ set  $X$ would
  have to hit all of the $k+1$ pairwise disjoint sets $e_i - C$ for $i
  \in  \{1,\dots,k+1\}$, which is impossible. 
\end{proof}

Recall that the idea behind the parallel computation of a kernel for
the hitting set problem is to repeatedly remove all
sunflowers from $H$, each time perhaps leaving a
manageable number of hyperedges~--~and after $d$ rounds, no hyperedges will
remain. We use the following notation for the ``removal'' operation:
For two hypergraphs $H = (V,E)$ and $H' = (V,E')$ let $H \ominus H' =
\bigl(V, \{\,e \in E \mid \forall e' \in E' \colon e' \not\subseteq
e\,\}\bigr)$, that is, we remove all hyperedges from $H$ that contain a
hyperedge of $H'$. Thus, $H \ominus H_1$ is the set of all hyperedges
in $H$ that are not involved in any sunflower of size at least $k+1$
since we remove all edges that contain a core.

The following theorem shows that the repeated removing operation only
leaves behind a ``small'' number of hyperedges. We formulate the
theorem for arbitrary matryoshka sequences (we will need this later
on), but it is best to think of the $M_i$ as the sets~$H_i$.

\begin{theorem}[Kernel Theorem]\label{theorem-kernel}
  Let $(M_0,\dots,M_{d(H)})$ be a matryoshka sequence for $H$
  and~$k$. Let 
  {\(
    K = (M_0 \ominus M_1) \cup (M_1 \ominus M_2) \cup (M_2 \ominus
    M_3) \cup \dots \cup (M_{d(H)-1} \ominus M_{d(H)}) \cup M_{d(H)}.
  \)}
  {\begin{enumerate}
  \item 
    Then $K$ has at most $\sum_{i=0}^{d(H)} k^i i!$ hyperedges and
  \item
    $H$ and $K$ have the same size-$k$ hitting sets.
  \end{enumerate}}
\end{theorem}


\begin{proof}
  For the first item, fix an~$i$ and consider $M_i \ominus
  M_{i+1}$. We claim that these ``remaining hyperedges'' cannot contain
  a sunflower $\{e_1,\dots,e_{k+1}\}$: If it did, by the third
  property of a matryoshka sequence the sunflower's core would be an
  element of $M_{i+1}$ and, thus, none of the $e_i$ would 
  be in $M_i \ominus M_{i+1}$. By the Sunflower Lemma and the fact
  that $d(M_i \ominus M_{i+1}) \le d(M_i) \le d(H)-i$ by the second
  property of a matryoshka sequence, we get that $M_i
  \ominus M_{i+1}$ cannot have more than $k^{d(H)-i}\cdot (d(H)-i)!$
  hyperedges. This means that the union $K$ cannot have more
  hyperedges than the sum of these numbers for $i \in
  \{1,\dots,d(H)\}$ plus the number of hyperedges
  in~$M_{d(H)}$. However, by the second property we have $d(M_{d(H)})
  \le 0$ and, thus, this hypergraph can contain at most one edge (the
  empty edge). We account for this single edge by the value $k^0 0! =
  1$ for $i=0$ in the sum $\sum_{i=0}^{d(H)} k^i i!$.

  For the second item, we make a simple observation: Let $A$, $B$, and
  $C$ be hypergraphs (all with the same vertex set~$V$ as always) such
  that every size-$k$ hitting set of $A \cup B$ is also a hitting set of
  $C$. We claim that
  \begin{align}
    \text{$A \cup B$ and $A \cup (B\ominus C) \cup C$ have the same
    size-$k$ hitting sets.} \tag{$*$}
  \end{align}
  To see this, first consider a size-$k$ hitting set~$X$ of $A \cup
  B$. Trivially, $X$ is also a hitting set of $A
  \cup (B \ominus C) \subseteq A \cup B$ and $X$ is also hitting set of~$C$ by
  assumption. Now, second, consider a size-$k$ 
  hitting set~$X$ of $A \cup (B\ominus C) \cup C$. Trivially, $X$ hits
  all of~$A$ as well as all hyperedges in~$B$ that are in $B \ominus
  C$, so consider an edge $e \in B - (B\ominus C)$. By definition,
  this means that $e$ must be a superset of some $e' \in E(C)$ and $X$
  was a hitting set for~$C$ and, thus, hits $e'$ and therefore also~$e$.

  Let us now prove the second item by proving by induction on~$i$ that
  $H$ and
  \begin{align*}
    (M_0 \ominus M_1) \cup (M_1 \ominus M_2) \cup (M_2 \ominus
    M_3) \cup \dots \cup (M_{i-1} \ominus M_{i}) \cup M_{i}.    
  \end{align*}
  have the same size-$k$ hitting sets. The base case $i=0$ is true by
  the first property of a matryoshka sequence. For the
  inductive step from $i$ to $i+1$, let $A = (M_0 \ominus M_1) \cup
   \dots \cup (M_{i-1} \ominus M_{i})$ and $B = M_i$ and $C =
  M_{i+1}$. By the inductive assumption, $H$ and $A \cup B$ have the
  same size-$k$ hitting sets. The fourth property of a matryoshka
  sequence now implies that every size-$k$ hitting set of $A \cup B$
  is also a hitting set of~$C$. By ($*$) we 
  get that $A \cup B$ has the same size-$k$ hitting sets as $A \cup (B
  \ominus C) \cup C$. Thus, $H$ and $A \cup (B \ominus C) \cup C$ have
  the same size-$k$ hitting sets, which was exactly the inductive
  claim for~$i+1$.  
\end{proof}

Instantiating the theorem with
$(H_0,\dots,H_{d(H)})$ tells us that, if we can compute the elements of $K =
(H_0 \ominus H_1) \cup \dots \cup (H_{d(H)-1} \ominus H_{d(H)}) \cup
H_{d(H)}$ in parallel, we can compute a kernel for the hitting set
problem in parallel. Clearly, ``computing $K$'' essentially boils down
to ``computing the $H_i$'' in parallel. Thus, the real question, which
we address next, is how quickly and easily we can compute the
hypergraphs~$H_i$.

At this point, we briefly need to address some technical issues
concerning the coding of hypergraphs. For our purposes, it is largely
a matter of taste how the input hypergraph $H_0$ is encoded, but the
encoding of the later graphs~$H_i$ becomes important in the context of parallel
constant-time computations. We consider $H = (V,E)$ fixed and encoded
using, for instance, an incidence matrix (having $|V|$ columns and
$|E|$ rows). We encode a \emph{refinement of $H$,} that is, a
hypergraph $H' = (V, E')$ with the property that each $e' \in E'$ is a
subset of some $e \in E$, using a matrix of $2^{d(H)}$ columns and $|E|$
rows. There is a column for each of the at most $2^{d(H)}$ possible
subsets of an edge $e \in E$ and the entry at the column for a given
row is $1$ if this subset is an element of $E'$; otherwise it is $0$.
Let us call this the \emph{refinement matrix enconding} of
hypergraph~$H'$ (with respect to the fixed input hypergraph~$H$).


\begin{lemma}[Computing Cores in Constant Depth]\label{lemma-ac0-hi}
  For each $d$ and $i$ there  is a \textsc{dlogtime}-uniform family of
  $\Class{AC}$-circuits that
 {\begin{enumerate}
  \item on input of the incidence matrix of a hypergraph $H$ with $d(H) \le
    d$, a number $k$, and the refinement matrix encoding of the
    hypergraph $H_i$,
  \item outputs the refinement matrix encoding of $H_{i+1}$,
  \item has constant depth, and
  \item has size $f(k,d) \cdot |V|^{O(1)} |E|^{O(1)}$ where $f$ is
    some computable function.
  \end{enumerate}}
\end{lemma}
\begin{proof}
  By the definition of refinement matrix encodings, it suffices to
  show how we can decide for a single set $e' \subseteq e \in E$
  whether it is a hyperedge of $H_{i+1}$ or not using a circuit of
  constant depth and size  $f(k,d) \cdot |V|^{O(1)} |E|^{O(1)}$ -- it
  then follows that we can do this for all possible $e'$ in parallel
  without increasing the circuit depth at all and without increasing
  the circuit size by more than $2^d |E|$.

  By definition, $e'$ is a hyperedge of $H_{i+1}$ if it is a $k$-core
  of $H_i = (V,E_i)$. Thus, ``all'' we need to test is whether there
  are petals $p_1,\dots, p_{k+1} \in E_i$ that form a
  sunflower in $H_i$ with core~$e'$. This question can be answered
  very quickly in parallel using color coding as follows: The vertices
  from the Color Coding Lemma (Fact~\ref{fact-cc}) that we ``have in
  mind'' are the vertices in the petals and the color ``we have in
  mind for them'' is $i$ for all vertices in $p_i - e'$. Formally, we use Fact~\ref{fact-cc} to obtain a set $\Lambda$ of
  colorings $\lambda\colon V \to \{1,\dots,k+1\}$ and check whether for one
  of these colorings for each $j \in \{1,\dots,k+1\}$ there exists a
  hyperedge $p_j \in E_i$ with the properties that (1) $p_j \supseteq
  e'$ and (2) all vertices in $p_j - e'$ have the color~$j$. Clearly,
  such a coloring together with the hyperedges $p_1$ to $p_{k+1}$
  exists if, and only if, $e'$ is the core of a sunflower of size
  $k+1$ in~$H_i$. Even more importantly, Fact~\ref{fact-cc} provides
  us with such a coloring~$\lambda$ -- if it
  exists -- via a circuit of constant depth and size at most some
  polynomial in the number $|V|$ of vertices times a computable
  function $f(k,d)$ depending on the number $k+1$ of colors and the
  maximum number $d(k+1)$ of vertices in the sunflower (for which we
  ``had colors in mind''). 
\end{proof}

The lemma tells us that once we have computed some $H_i$, we can
compute the next $H_{i+1}$ using only constant additional depth and
using $f(k,d) \cdot |V|^{O(1)} |E|^{O(1)}$ additional size. Since
$H_i \ominus H_{i+1}$ can easily be computed from $H_i$ and $H_{i+1}$
in constant depth, we get:

\begin{theorem}[Depth-$O(d)$ Kernelization Algorithm, \cite{ChenFH2017}]\label{thm-kernel-1}
  For each~$d$ there is a \textsc{dlogtime}-uniform family of
  $\Class{AC}$-circuits that
  \begin{enumerate}
  \item on input of a hypergraph $H$ with $d(H) \le
    d$ and a number $k$
  \item outputs a hypergraph $K$ having the same
    size-$k$ hitting sets as $H$ and having at most $\sum_{i=0}^{d(H)}
    k^i i!$ hyperedges, 
  \item has depth $O(d)$,
  \item and has size $f(k,d) \cdot |V|^{O(1)} |E|^{O(1)}$ where $f$ is
    some computable function.
  \end{enumerate}  
\end{theorem}


\section{Pseudo-Cores and Pseudo-Sunflowers}
\label{section-pseudo}

The parallel kernelization algorithm
described in the previous section has a depth that is linear in the
parameter $d$, the maximum size of 
any hyperedge in the input hypergraph. The reason for this linear
dependency was that, while we managed to reduce not just one but
all sunflowers in the hypergraph to their cores in parallel, we had to
repeat this ``reduce to core'' procedure $d$ times -- and each round
adds a constant number of layers to the circuit.

It is not obvious how this build-up of layers can be avoided. In the
following, we first explain why there are good
reasons to believe that the computation of the hypergraphs $H_i$
necessitates deeper and deeper circuits. Following this discussion, we
explain our proposal for side-stepping these difficulties: we replace
the hypergraphs $H_i$ by new hypergraphs $H'_i$ that are easier to
compute but still form a matryoshka sequence and -- hence -- can serve
as a replacement for the $H_i$ in the Kernel Theorem, Theorem~\ref{theorem-kernel}.

\paragraph*{The Difficulty: Cores of Cores Are Hard to Compute}

There are several reasons to believe that one cannot compute kernels
for the hitting set problem in constant depth using the repeated
sunflower-reduction-procedure. 
A first idea for reaching a constant depth is to apply the
reduction procedure only a constant number of times (instead of $d$
times). Indeed, it is not immediately clear that a ``core of cores''
is not already a core in the first round -- so do we actually need
more than \emph{one} round? Unfortunately, the answer is ``yes, we
do'': Figure~\ref{fig:example1} shows an example where $\{a,b\}$ is a
2-core of the 2-cores, but it is not a 2-core of the original
hypergraph. For a more complex example, where $d$ rounds are needed to
arrive at a constant size kernel, consider the trees $T^\ell_d$
(defined in detail later on) that are perfectly balanced trees of
depth $d$ with $\ell+1$ children per node for a number $\ell \ge k$ --
and now consider the hypergraph $H^d$ that has one hyperedge  
for each leaf of $T^\ell_d$ and this hyperedge contains all the nodes on the
path from the leaf to the root~$r$. Now, for $i>0$ we have
$\operatorname{\mathit k-cores}(H^i) = H^{i-1}$ and the latter
hypergraphs all have a size of at least the arbitrarily large~$\ell$
for~$i>1$. Thus, we need to apply the ``core of cores'' procedure at
least $d-1$ times before arriving at a hypergraph whose size depends
only on the parameter.  

A second, more promising idea is the observation that it might be
possible to somehow ``collapse'' two (and then, hopefully, all)
applications of the sunflower-reduction-procedure ``into a single
application.'' Unfortunately, we also run into a problem here, namely in
the ``collapsed color coding process.'' In essence, color coding is
great at ensuring that certain vertex sets are disjoint (namely those
vertex sets that receive different colors), but fails at enforcing
that the same vertices are used in different hyperedges -- which is
exactly what is needed when the definition of some $H_i$ refers to
$H_{i-1}$, which in turn refers to some $H_{i-2}$.


These problems with avoiding the build-up of additional
layers with rising~$d$ have led Chen et al.~\cite{ChenFH2017} to the
conjecture that the build-up is unavoidable and that all parallel
kernelization algorithms for $\PLang[k,d]{hitting-set}$ have a runtime
that is linear in~$d$. We agree with Chen et al.\ in their assessment
that the computation of the $H_i$ presumably necessitates a linear
circuit depth -- but, nevertheless, we will refute their conjecture in
the following.

\paragraph*{The Solution: Pseudo-Cores As a Replacement For Cores}

Our idea is \emph{not} to compute the sets $H_i$ (we do not see how
this can be done in constant time), but to compute hypergraphs $H'_i$
with rather similar properties (formally, they will form matryoshka
sequences as well) that we \emph{can} compute in constant time for all $d$
and~$i$.
We introduce a new notion of \emph{$k$-pseudo-cores of level~$i$} and
$H'_i$ will be the hypergraph whose edges are the $k$-pseudo-cores of
level~$i$. Crucially, the definition of $H'_i$ (only) refers directly
to the original input graph~$H$ and its hyperedges can be obtained
from $H$ directly using color coding. At the same time, the $H'_i$
will form a matryoshka sequence and, hence, just as for the $H_i$, the
core of any sunflower of $H'_{i-1}$ must already be present in $H'_i$.

The definition of pseudo-cores is somewhat technical.
We will, however, show that all cores are pseudo-cores of level 1, cores of
cores are pseudo-cores of level~2, and so on. The reverse implication
does not hold (for instance, pseudo-cores of level~2 need not be cores
of cores). 
For a ``level'' $L$ and a number~$k$, let $T_L^k$ denote the rooted
tree in which all leafs are at the same depth~$L$ and all inner nodes
have exactly $k+1$ children. The root of $T_L^k$ will always be called
$r$ in the following. Thus, $T_1^k$ is just a star consisting of $r$
and its $k+1$ children, while in $T_2^k$ each of the $k+1$ children of
$r$ has $k+1$ new children, leading to $(k+1)^2$ leafs in total. For
each $l\in \operatorname{leafs}(T_L^k) = \{\,l \mid \text{$l$ is a leaf
  of $T_L^k$}\,\}$ there is a unique path $(l^0,l^1,\dots,l^L)$ from
$l^0 = r$ to $l^L = l$.  An example for the following definition is
shown in  Figure~\ref{fig:example-pseudo}.

\begin{definition}[Pseudo-Sunflowers and Pseudo-Cores]\label{def-pseudo-core}
  Let $H = (V,E)$ be a hypergraph and let $L$ and $k$ be fixed. A set
  $C \subseteq V$ is called a \emph{$k$-pseudo-core of level~$L$ in
    $H$} if there exists a mapping $S \colon
  \operatorname{leafs}(T_L^k) \times \{0,1,\dots,L\} \to 2^V$, called a
  \emph{$T_L^k$-pseudo-sunflower for $H$ with pseudo-core~$C$}, such that for all $l, m \in
  \operatorname{leafs}(T_L^k)$ with $l\neq m$ we have:  
  \begin{enumerate}
  \item $S(l,0) = C$.
  \item $S(l,0) \cup S(l,1) \cup \dots \cup S(l,L) \in E$ and let us
    write $S(l)$ for this hyperedge.
  \item $S(l,i) \cap S(l,j) = \emptyset$ for $0\le i < j \le L$, but $S(l,i)
    \neq \emptyset$ for $i \in \{1,\dots,L\}$.
  \item Let $z \in \{1,\dots, L\}$ be the smallest number such that
    $l^z \neq m^z$, that is, $z$ is the depth where the path from $r$
    to $l$ and the path from $r$ to $m$ diverge for the first
    time. Then $S(l,z) \cap S(m,z) =\emptyset$ must hold.
  \end{enumerate}
\end{definition}

\begin{figure}[htpb]
  \begin{tikzpicture}

    \node [below right] at (0mm,4.8cm) {$H$};
    
    \vertices
    \scoped[r/.style=red,g/.style=green!50!black,b/.style=blue]{
      \hyperedgesh
    }

  \end{tikzpicture}
  \hfill  
  \begin{tikzpicture}[every label/.style={font=\small},yscale=1.2]
    \node [] at (0,1.5) {$T_2^2$};
    
    \node [dot,label=left:$r$] (r) at (0,0) {};
    \node [dot,label=left:$c_1$] (c1) at (1,3em) {};
    \node [dot,label=above left:$c_2$] (c2) at (1,0) {};
    \node [dot,label=left:$c_3$] (c3) at (1,-3em) {};

    \node [dot,label=above:$l$] (c11) at (2,4em) {};
    \node [dot] (c12) at (2,3em) {};
    \node [dot] (c13) at (2,2em) {};

    \node [dot] (c21) at (2,1em) {};
    \node [dot] (c22) at (2,0em) {};
    \node [dot] (c23) at (2,-1em) {};

    \node [dot] (c31) at (2,-2em) {};
    \node [dot] (c32) at (2,-3em) {};
    \node [dot] (c33) at (2,-4em) {};

    \graph [use existing nodes] {
      r -- {
        c1 -- {c11, c12, c13},
        c2 -- {c21, c22, c23},
        c3 -- {c31, c32, c33}
      }
    };

    \foreach \l/\colo[count=\i] in {
      {{a,b},{c,f},{u,v,w}}/red,
      {{a,b},{c,g},{r,s,t,m}}/red,
      {{a,b},{c,h,o},{p,q,l,e}}/red,
      {{a,b},{d,i},{o,r,u}}/blue,
      {{a,b},{d,j},{p,s,v}}/blue,
      {{a,b},{d},{k,q,t,w}}/blue,
      {{a,b},{e},{l}}/green!50!black,
      {{a,b},{e},{m}}/green!50!black,
      {{a,b},{e},{n}}/green!50!black%
    } {
      \foreach \entry[count=\j] in \l {
        \node [anchor=mid west,color=\colo] at (1cm+\j cm*1.5,5em-\i em) {$\{\entry\}$};
      }
    }

    \node [anchor=mid west] at (2.5cm, 5.5em) {$S(l,0)$};
    \node [anchor=mid west] at (4cm, 5.5em) {$S(l,1)$};
    \node [anchor=mid west] at (5.5cm, 5.5em) {$S(l,2)$};
    
  \end{tikzpicture}
  
  \caption{A $T_2^2$-pseudo-sunflower~$S$ for the level~2 pseudo-core
    $\{a,b\}$ in the hypergraph~$H$. The four properties of 
    pseudo-sunflowers hold: In ``column $S(l,0)$'' we always
    have the  pseudo-core, the union of each row is a hyperedge, the
    sets in a row form a partition of this hyperedge, 
    and -- most importantly~-- we have the disjointness property at each
    ``branch'' of the tree. This property requires that for column
    $S(l,1)$ the sets of all red vertices, of all blue vertices, and of
    all green vertices are pairwise disjoint; whereas for column
    $S(l,2)$ it requires that the three red sets are pairwise
    disjoint, 
    likewise for the three blue sets, and the three green
    sets. However, it is permissible (and the case) that a red vertex in
    the third column is the same as green vertex in the third or the
    second column.} 
  \label{fig:example-pseudo}
\end{figure}
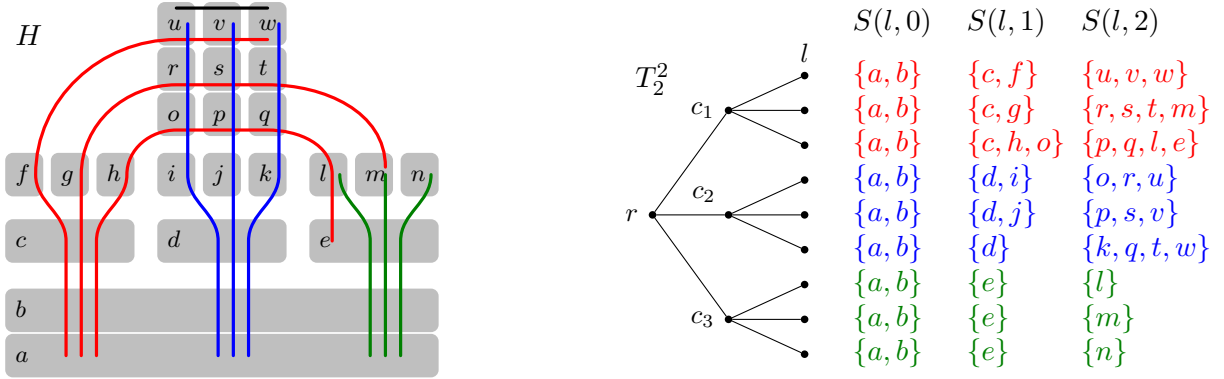


\begin{definition}\label{definition-pseudo}
  For a hypergraph $H = (V,E)$ and numbers $k$ and $i \ge 1$ let $H'_i
  = \bigl(V, \{C \mid \text{$C$ is a $k$-pseudo-core of level~$i$
    of~$H$}\}\bigr)$ and let $H'_0 = H$.
\end{definition}
To get some intuition, let us have a closer look at $H'_1$. As the
following lemma shows, pseudo-cores and cores are still \emph{very}
closely related at this first level -- while for larger levels, we no
longer have $H_i = H'_i$, but only $H_i \subseteq H'_i$.

\begin{lemma}\label{lemma-ph1}
  Let $H$ be a hypergraph and $k$ a number. Then  $H_1 = H'_1$.
\end{lemma}
\begin{proof}
  Consider a $k$-pseudo-core of~$H$ of level~$1$. The tree
  $T_1^k$ consists of a root~$r$ with 
  leafs $l_1$ to $l_{k+1}$. Consider a $T_1^k$-pseudo-sunflower $S$ 
  and let us fix some leaf $l$ of~$T_1^k$. The pseudo-sunflower must  map
  $(l,0)$ to the pseudo-core $C$ and $(l,1)$ to a set of
  vertices that is disjoint from~$C$. This means that $S(l) =
  S(l,0) \cup S(l,1) = C \cup S(l,1) \in E$ is a hyperedge in $H$ that
  contains the (pseudo)core~$C$. Furthermore, for any two different
  leafs $l$ and $m$  we have $S(l,1) \cap S(m,1) = \emptyset$ or,
  equivalently, $S(l) \cap S(m) = C$. This shows that $\{S(l_1),\dots,
  S(l_{k+1})\}$ is a sunflower with core~$C$. For the other direction,
  given a sunflower $\{e_1,\dots,e_{k+1}\}$ of size $k+1$ in~$H$ 
  with core~$C$, the $T_1^k$-pseudo-sun\-flower $S$ given by $S(l_i,0) =
  C$ and $S(l_i,1) = e_i - C$, where the $l_i$ are the $k+1$ leafs
  of $T_1^k$, witnesses that $C$ is also $k$-pseudo-core of level~1 of~$H$. 
\end{proof}

\section{The Constant-Depth Kernelization}
\label{section-algorithm}

We show that hitting set kernels can be
computed in constant depth in two steps:
\begin{enumerate}
\item We show that $(H'_0,\dots,H'_{d(H)})$ is a matryoshka sequence.
\item We show that all $H'_i$ can be computed by a constant depth
  circuit whose depth is independent of both $k$ and $d(H)$.
\end{enumerate}
By the Kernel Theorem, Theorem~\ref{theorem-kernel}, taken together,
these two items yield the desired kernelization algorithm.

\paragraph*{Step 1: Pseudo-Cores Form Matryoshka Sequences.}

Our first aim is to show the following theorem, which is an analogue
of Lemma~\ref{lemma-core-properties} for pseudo-cores:

\begin{theorem}\label{theorem-pseudo}
  For every hypergraph $H$ and number $k$, the sequence
  $(H'_0,\dots,H'_{d(H)})$ from Definition~\ref{definition-pseudo}
  is a matryoshka sequence for~$H$ and~$k$.  
\end{theorem}
The proof consists of four lemmas, one for each of four properties of
a matryoshka sequence:

\begin{lemma}\label{lemma-p1}
  $H'_0 = H$.
\end{lemma}

\begin{proof}
  By definition.
\end{proof}

\begin{lemma}\label{lemma-p2}
  $d(H'_L) \le d(H) - L$ holds for all $L \in \{0,\dots,d(H)\}$.
\end{lemma}

\begin{proof}
  For every leaf $l$ we have $S(l) =  S(l,0) \mathbin{\dot\cup} S(l,1)
  \mathbin{\dot\cup} \cdots \mathbin{\dot\cup} S(l,L)$ and all
  $S(l,i)$ for $i \in \{1,\dots,L\}$ are non-empty sets. This implies
  that $|S(l,0)| \le |S(l)| - L \le d(H) - L$.
\end{proof}

\begin{lemma}\label{lemma-p3}
  $\operatorname{\mathit k-cores}(H'_L) \subseteq H'_{L+1}$ holds for
  all $L \in \{0,\dots,d(H)\}$. 
\end{lemma}
\begin{proof}
  We show the claim by induction on~$L$. The base case $L=0$ was already handled in Lemma~\ref{lemma-ph1}. For
  larger~$L$, let $e$ be a $k$-core of $H'_L$. Then there is a
  sunflower $\{e_1,\dots,e_{k+1}\} \subseteq E(H'_L)$ with
  core~$e$ and $|e| \le d(H) - L - 1$ since all $e_i$ have the property
  $|e_i| \le d(H) - L$ by definition and since a core is always
  smaller than its largest petal. For each $j \in \{1,\dots,\penalty0k+1\}$
  there must now be a $T_L^k$-pseudo-sunflower $S_j$ with
  pseudo-core~$e_j$. From these $k+1$ different pseudo-sunflowers, we
  construct a $T_{L+1}^k$-pseudo-sunflower $S$ whose pseudo-core is~$e$ as
  follows: In the tree $T_{L+1}^k$, let $c_1$ to $c_{k+1}$ be the children of
  the root~$r$. For a leaf~$l$ of $T_{L+1}^k$, let $l^1 = c_j$ be the
  child of~$r$ on the path from $r$ to~$l$ and let us view $l$ also as
  a leaf in the tree $T_L^k$ rooted at $c_j$. We define $S$ as
  follows:
  \begin{enumerate}
  \item $S(l,0) = e$.
  \item $S(l,1) = e_j - e$.
  \item $S(l,i) = S_j(l,i-1)$ for $i \in \{2,\dots,L+1\}$.
  \end{enumerate}
  It remains to show that the mapping $S$ defined in this way
  satisfies the four properties of a pseudo-sunflower. Consider any
  two leafs $l$ and $m$ of $T_{L+1}^k$:
  \begin{enumerate}
  \item By definition, $S(l,0) = e$ and, thus, $e$ is the
    pseudo-core of $S$.
  \item $S(l,0) \cup S(l,1) \cup \dots \cup S(l,L+1) = e \cup (e_j - e)
    \cup S_j(l,1) \cup \dots \cup S_l(l,L) = e \cup (e_j - e)
    \cup (S_j(l)-e_j)$. We know that $S_j(l) \in E$ holds (since $S_j$
    is a pseudo-sunflower), that $e_j \subseteq S_j(l)$ holds (since
    $e_j$ is the pseudo-core of~$S_j$), and that $e \subseteq e_j$
    holds (since $e$ is the core of a sunflower that has $e_j$ as one
    of its petals). This implies $S(l) = S_j(l)$ and the latter is an
    element of~$E$.
  \item Clearly, $S(l,0)= e$ and $S(l,1) = e_j-e$ are disjoint and
    $S(l,1)$ has size at least $1$ since $e_j$ is not a subset
    of~$e$. The other $S(l,i)$ are also disjoint from one another since the
    $S_j(l,i-1)$ are disjoint among one another, and they are also
    disjoint from $S(l,0)$ and $S(l,1)$ (since all
    $S_j(l,i-1)$ are disjoint from $S_j(l,0) = e_j$).
  \item Finally, let $z$ be the ``divergence depth'' of $l$ and $m$,
    that is, the smallest number with $l^z \neq m^z$. For $z=1$, we
    have $S(l,1) \cap S(m,1) = \emptyset$ since  $S(l,1) = e_j - e$
    and $S(m,1) = e_{j'} -e$ for some indices $j$ and $j'$ and since
    the intersection of the two sunflower petals $e_j$ and $e_{j'}$ is
    exactly~$e$. For $z>1$, the leafs $l$ and $m$ lie in the same tree
    $T_L^k$ rooted at some child $c_j$ of~$r$ and we have $S(l,z)
    \cap S(m,z) = S_j(l,z-1) \cap S_j(m,z-1)$ and the latter
    intersection is empty since $S_j$ is a pseudo-sunflower.\qedhere 
  \end{enumerate}
\end{proof}


\begin{lemma}\label{lemma-p4}
  Every size-$k$ hitting set of~$H$ is also a size-$k$ hitting set of~$H'_L$
  for all $L \in \{0,\dots,d(H)\}$.  
\end{lemma}
\begin{proof}
  We must show that every hitting set $X$ of $H$ with $|X| \le k$ is
  also a hitting set of $H'_L$. That is, we must show that every
  $k$-pseudo-core $C$ of level~$L$ contains at least one element
  of~$X$. Let $S$ be a $T_L^k$-pseudo-sunflower with pseudo-core~$C$.

  The following definition will be crucial in the following: We say
  that \emph{$X$ hits a node $n$ of $T_L^k$} if there is a leaf $l$ of
  $T_L^k$ such that $n = l^D$ (that is, $n$ is the $D$th node on the
  path $l^0,l^1,\dots,l^L$ from the root $r = l^0$ to $l = l^L$) and
  $X \cap (S(l,0) \cup S(l,1) \cup \dots \cup S(l,D)) \neq \emptyset$.

  First, observe that $X$ hits every leaf of $T_L^k$ since, for them,
  $D = L$ and $S(l,0) \cup \dots \cup \dots S(l,L) = S(l)$ is a
  hyperedge of $H$ and, hence, gets hit by~$X$.

  Second, we claim that if $X$ hits all children $c_1,\dots,c_{k+1}$
  of a node~$n$ of $T_L^k$, then $X$ also hits~$n$. Let $n$ be at
  depth $D$, so the children are at depth~$D+1$. By definition of
  ``being hit by $X$,'' for each child $c_i$ of~$n$ there must be a leaf
  $l_i$ such that
  \begin{align}
    X \cap (S(l_i,0) \cup \dots \cup S(l_i,D+1)) \neq \emptyset.\label{eq-1}
  \end{align}
  We claim that at least one of the $l_i$
  also witnesses that $X$ hits~$n$. Otherwise, since all $l_i$ are in
  the subtree rooted at $n$, for all of them we would have
  \begin{align}
    X \cap (S(l_i,0) \cup \dots \cup S(l_i,D)) =
    \emptyset.\label{eq-2}
  \end{align}
  From \eqref{eq-1} and \eqref{eq-2} we immediately get that
  \begin{align}
    X \cap S(l_i,D+1) \neq \emptyset \text{ must hold for all $i \in
    \{1,\dots,k+1\}$}.\label{eq-3}
  \end{align}
  Now, for any two
  different leafs $l_i$ and $l_j$ consider the two paths from the root to
  them. These paths will be identical exactly up to the node~$n$ and
  will then split into a path via the child $c_i$ and a path via the
  child $c_j$. Now, in this situation the fourth property of
  pseudo-sunflowers tells us that $S(l_i,D+1) \cap S(l_j,D+1) =
  \emptyset$ must hold. In other words, the $k+1$ many sets
  $S(l_i,D+1)$ in equation~\eqref{eq-3} \emph{are pairwise disjoint.}
  However, this means that the size-$k$ set~$X$ cannot contain one
  element of each of them. Thus, our assumption that $X$ does not hit
  $n$ has lead us to a contradiction.

  Third, we claim that $X$ hits the root of~$T_L^k$. This follows
  easily from the first two claims since $X$ hits all leafs of $T_L^k$ and
  whenever it hits all children of a node, it also hits the
  node. Clearly, this implies that $X$ hits all nodes, including the
  root.

  Now, we are done since ``$X$ hits the root'' means that $X \cap
  S(l,0) \neq \emptyset$ holds for at least one leaf~$l$ and $S(l,0) =
  C$. Thus, $X \cap C \neq \emptyset$, which was the claim.
\end{proof}

\paragraph*{Step 2: Pseudo-Cores Can Be Computed in Constant Depth.}

Theorem~\ref{theorem-pseudo} states that the hypergraphs $H'_i$ form
a matryoshka sequence and, thus, the Kernel Theorem tells us that
the following hypergraph is a kernel
for the hitting set problem:
\(
  K = (H'_0 \ominus H'_1) \cup (H'_1 \ominus H'_2) \cup \dots \cup (H'_{d(H)-1} \ominus
  H'_{d(H)}) \cup H'_{d(H)}. 
\)
Of course, the whole effort that went into the definition of the
$H'_i$ and the proof of the matryoshka properties would be for
nothing, if the $H'_i$ were not easier to compute than the~$H_i$.

This is exactly what we claim in the following theorem and prove in the
rest of this paper: It is an analogue of Lemma~\ref{lemma-ac0-hi} for
pseudo-cores. The crucial difference in 
the formulation is that, now, we no longer get $H'_{i-1}$ as input
when we compute $H'_i$, but rather we compute $H'_i$ ``directly'' from
the original graph~$H$.

\begin{theorem}[Computing Pseudo-Cores in Constant Depth]\label{thm-ac0-pseudo}
  There is a \textsc{dlogtime}-uniform family of $\Class{AC}$-circuits that
  {\begin{enumerate}
  \item on input of the incidence matrix of a hypergraph $H = (V,E)$
    and numbers~$k$ and~$L$,
  \item outputs the refinement matrix encoding of $H'_L$,
  \item has constant depth (in particular, it is independent of
    $|V|$, $|E|$, $d(H)$, $k$, and $L$), and
  \item has size $f(k,d(H)) \cdot |V|^{O(1)} |E|^{O(1)}$ where $f$ is
    some computable function.
  \end{enumerate}}
\end{theorem}

To compute the encoding of~$H'_L$,
we can consider all candidate pseudo-cores in parallel. Thus, proving
the theorem boils down to deciding for a subset $C \subseteq V$
whether there exists a $T_L^k$-pseudo-sunflower $S$ of~$H$ whose
pseudo-core is~$C$. Of course, we wish to use color coding for
this and our definition of pseudo-cores and pseudo-sunflowers was
carefully crafted so that it includes only requirements of the form
``these parts of these hyperedges must be disjoint'' (and not -- as is
necessary for describing cores of cores -- statements like ``these 
hyperedges must \emph{share} the vertices that form petals'').
Unfortunately, while we no longer need to \emph{ensure} that certain
parts of different hyperedges are identical, we must be careful
that we do not inadvertently \emph{forbid} vertices to be the same
across hyperedges when we ``do not care whether they are
the same'':
\begin{example}\label{example-ex1}
  Suppose we wish to find two disjoint hyperedges $e_1 = \{v_1,v_2,v_3\}$
  and $e_2 = \{v_4,v_5,v_6\}$ in a hypergraph~$H$ plus another
  hyperedge $e_3 = \{x,y\}$ such that $x \notin e_1 \cup e_2$, but do
  not care whether $y \in e_1\cup e_2$ holds or not. We can easily
  enforce the disjointness properties 
  by coloring $v_1$ to $v_6$ using colors $1$ to~$6$ and $x$ using
  color~$7$. However, how should we color $y$ for which 
  \emph{we do not care about disjointness} (at least with respect to
  $e_1$ and $e_2$)? Fixing any of the colors $1$ to~$3$ for~$y$ or  
  any of the colors $4$ to~$6$ (or, for that matter, any other color)
  would be wrong, since this would enforce either $y \notin e_2$ or $y
  \notin e_1$ (or both).  
\end{example}
Fortunately, there is a way out of the dilemma: we consider
all feasible colors~$y$ could get in parallel. To formalize this
``trick'', we define a 
technical problem in which an undirected graph~$G$ is used to
specify which vertices in hyperedges of a hypergraph~$H$ should be 
different. As is customary, a \emph{proper
  coloring} of an undirected graph $G = (U,F)$   
is a mapping $c \colon U \to C$ to some set $C$ of colors with $c(u)
\neq c(v)$ for all $\{u,v\} \in F$. Let us write $f[X] = \{f(x) \mid x
\in X\}$ for the image of a  
set $X$ under a function~$f$. For an example instance see
Figure~\ref{fig:example-restricted}. 

\begin{problem}{$\PLang[G]{restricted-coloring}$}\label{technical-problem}
 \begin{description}\parskip0pt\itemsep0pt
    \item[Instance:] A hypergraph $H=(V,E)$ and an undirected graph $G =
    (U,F)$ together with a partition $U = U_1 \mathbin{\dot\cup}
    \cdots \mathbin{\dot\cup} U_m$ of~$U$.  
    \item[Parameter:] $|G|$
    \item[Question:] Is there a proper coloring $c \colon U \to V$ of~$G$ such
    that $c[U_i] \in E$ holds for all $i \in \{1,\dots,m\}$?
  \end{description}
\end{problem}

\begin{figure}[htpb]
  \begin{tikzpicture}[baseline]

    \node [below right] at (0mm,4.8cm) {$H'$};
    
    \verticesup
    \scoped[r/.style=,g/.style=,b/.style=]{
      \hyperedgeshup
    }

  \end{tikzpicture}
  \hfill  
  \begin{tikzpicture}[baseline]
    \node [below right] at (10mm,4.8cm) {$G$};

    \foreach \l in {1,...,9} {
      \foreach \i in {1,2} {
        \foreach \d in {1,...,4} {
          \node [dot] (\l\i\d) at (\i*1.5cm+\d*3mm,14em-\l*1.5em) {};
        }
      }
    }

    \node [above right] at ([xshift=4cm,yshift=2mm]111) {\emph{Proper coloring}};

    \draw [densely dashed,rounded corners=1mm] ([shift={(-2mm,2mm)}]111) rectangle ([shift={(2mm,-2mm)}]314);
    \draw [densely dashed,rounded corners=1mm] ([shift={(-2mm,2mm)}]411) rectangle ([shift={(2mm,-2mm)}]614);
    \draw [densely dashed,rounded corners=1mm] ([shift={(-2mm,2mm)}]711) rectangle ([shift={(2mm,-2mm)}]914);

    \coordinate (x) at ([xshift=-2mm]211);
    \coordinate (y) at ([xshift=-2mm]511);
    \coordinate (z) at ([xshift=-2mm]811);

    \draw [very thick] (x) to[bend right] (y);
    \draw [very thick] (y) to[bend right] (z);
    \draw [very thick] (x) to[bend right] (z);
    
    \foreach \l in {1,...,9} {
      \foreach \d in {1,...,4} {
        \foreach \e in {1,...,4} {
          \draw [thin] (\l1\d) to [bend left=20] (\l2\e);
        }
      }
    }
    
    \foreach \l/\m in {1/2,2/3,4/5,5/6,7/8,8/9} {
      \foreach \d in {1,...,4} {
        \foreach \e in {1,...,4} {
          \draw [thin] (\l2\d) -- (\m2\e);
        }
      }
    } 
    \foreach \l/\m in {1/3,4/6,7/9} {
      \foreach \d in {1,...,4} {
        \foreach \e in {1,...,4} {
          \draw [thin] (\l2\d) to[bend left=10] (\m2\e);
        }
      }
    } 

    \draw[thick,black!50] (5.5,-0.25) -- (5.5,4.75);
    
    \foreach \l [count=\i] in {
      {{c,f,f,f},{u,v,w,w}},
      {{c,g,g,g},{r,s,t,m}},
      {{c,h,o,o},{p,q,l,e}},
      {{d,i,i,i},{o,r,u,u}},
      {{d,j,j,j},{p,s,v,v}},
      {{d,d,d,d},{k,q,t,w}},
      {{e,e,e,e},{l,l,l,l}},
      {{e,e,e,e},{m,m,m,m}},
      {{e,e,e,e},{n,n,n,n}}%
    }
    {
      \foreach \c [count=\j] in \l {
        \foreach \s [count=\k] in \c {
          \node [font=\footnotesize] at ([xshift=4cm]\i\j\k) {$\s$};
        }
      }
    }
  \end{tikzpicture}
  \caption{An instance of $\PLang[G]{restricted-coloring}$ consisting
    of a hypergraph~$H'$ and a graph~$G$ (a thick edge connecting two
    areas with dashed borders indicates that there is an edge between
    each vertex of the first area and each vertex of the second area;
    thus, in the example, each thick edge corresponds to
    $12 \cdot 12 = 144$ edges). This instance is the one resulting
    from the reduction described in the proof of
    Theorem~\ref{thm-ac0-pseudo} for $L=2$, the hypergraph~$H$ from
    Figure~\ref{fig:example1}, and the core $\{a,b\}$ (except that we
    use only four vertices in $G$ per set $S(l,i)$ instead of
    $d=9$). A proper coloring is shown right (the table indicates the
    values $c(u) \in V(H')$ for the corresponding vertices $u$ of~$G$).  }
  \label{fig:example-restricted}
\end{figure}
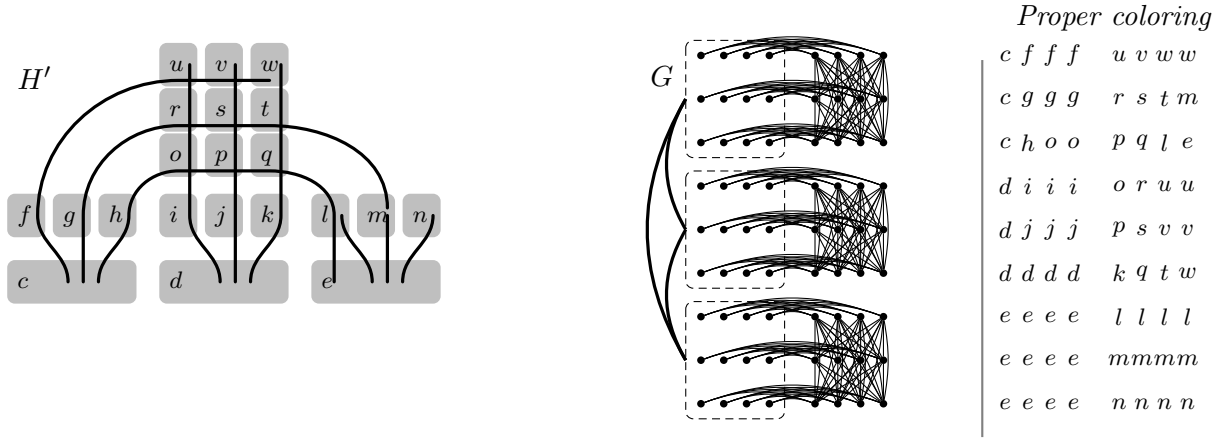

\begin{lemma}\label{lemma-restricted-cc}
  The problem $\PLang[G]{restricted-coloring}$ can be
  solved by a \textsc{dlogtime}-uniform family of
  $\Class{AC}$-circuits of constant depth and size $f(|G|) |V|^{O(1)}
  |E|^{O(1)}$ for some computable function~$f$.
\end{lemma}

\begin{proof}
  Let $G= (U,F)$ and $H = (V,E)$ be given as input. We assume that
  $|U| \le |V|$ holds since, otherwise, the number of possible
  mappings $c$ that must be checked only depends on the
  parameter~$|G|$ and, thus, they can easily be checked in parallel in
  constant depth (any function can be computed in constant depth when
  circuit size is not an issue).  

  The objective is, of course, to use color coding to find the
  mapping~$c$. Towards this aim, we search for a mapping $d \colon V
  \to \{1,\dots,|U|\}$ with the following two Properties 1 and~2:
  \begin{enumerate}
  \item There is a proper coloring $c' \colon U \to \{1,\dots,|U|\}$ of~$G$
    such that 
  \item for each $i\in \{1,\dots,m\}$ there is a hyperedge $e_i \in E$
    with $|d[e_i]| = |e_i|$ and  $d[e_i] = c'[U_i]$. 
  \end{enumerate}

  The idea behind this search is the following: The graph $G$ imposes
  restrictions of the form ``for the edge $\{u,v\}$ the vertices
  $c(u)$ and $c(v)$ must be different.'' In principle, such
  restrictions are easy to enforce using color coding: use one color
  for $c(u)$ and a different color for $c(v)$. However, as explained
  in Example~\ref{example-ex1}, we may run into a problem when there
  is \emph{no} edge between two vertices $x$ and~$y$ since, then, we
  may not rule out that $c(x) = c(y)$ holds -- which we would rule
  out when two different colors are used. The solution offered
  following the example is to try out all possible ways in which we
  may assign colors to vertices for which we ``actually do not care''
  about their colors. These ``possible ways'' are modeled by the proper
  coloring~$c'$ from above.

  In detail, recall the situation of Example~\ref{example-ex1} where
  we searched for two disjoint hyperedges $e_1 = \{v_1,v_2,v_3\}$ and
  $e_2 = \{v_4,v_5,v_6\}$ and a hyperedge $e_3 = \{x,y\}$ with $x
  \notin e_1 \cup e_2$. This search can be modeled by a graph~$G$
  whose vertex set $U$ can be partitioned into $U_1 =
  \{u_1,u_2,u_3\}$, $U_2=\{u_4,u_5,u_6\}$, and $U_3 = \{u_7,u_8\}$ and
  where the edge set~$F$ is a clique on $U_1 \cup U_2$ (to ensure that
  the hyperedges $c[U_1]$ and $c[U_2]$ are disjoint and have size~$3$)
  and there are edges between $u_7$ and all elements of $U_1 \cup U_2$ (to
  ensure that $c(u_7) \notin c[U_1] \cup c[U_2]$). A proper coloring $c
  \colon U \to V$ might now map $c(u_i) = v_i$ for
  $i\in\{1,\dots,6\}$ and $c(u_7) = x$ and $c(u_8) = y$. All vertices
  in $\{v_1,\dots,v_6,x\}$ must be distinct, but $y$ must not
  necessarily be distinct from them~--~in fact, it could be any of them.
  In this situation, the different possible values of~$y$ give
  rise to different $c'$ and~$d$ (note that in all of these examples,
  $c'$ is a proper coloring of~$G$ and that $|d[e_i]| = |e_i|$ and
  $d[e_i] = c'[U_i]$ hold for $i\in\{1,2,3\}$):

  \begin{example}\label{example-y1}
    Suppose that $y=v_1$ holds, that is, $e_1 \cap e_3 = \{v_1\}$ and
    $e_2 \cap e_3 = \emptyset$. This situation is modeled by the 
    following functions $c'$ and~$d$: $c'(u_i) = i$ for
    $i\in\{1,\dots,7\}$ and $c'(u_8) = 1$; and $d(v_i) = i$ for
    $i\in\{1,\dots,6\}$ and $d(x) = 7$ and $d(v)$ can be arbitrary for
    $v\notin \{v_1,\dots,v_6,x\}$ (note that we do not need to define
    $d(y)$ since $d(v_1)$ is already defined and $y = v_1$).  
  \end{example}
  \begin{example}\label{example-y5}
    Suppose that $y=v_5$ holds, that is, $e_1 \cap e_3 = \emptyset$ and
    $e_2\cap e_3 = \{v_5\}$. Here, we can use almost identical
    functions $c'$ and~$d$ as in the previous example, except that
    $c'(u_8) = 5$.  
  \end{example}
  \begin{example}\label{example-y*}
    Suppose that $y
    \notin \{v_1,\dots,v_6,x\}$. Then we use $c(u_8) = 8$ and $d(y) =
    8$.
  \end{example}
  
  Let us now formally argue that the search for $d$ can be performed
  using color coding: 
  First, observe that the test ``there is a proper coloring $c'$'' can be
  performed in parallel by testing all possible colorings of~$G$
  (their number depends only on $|G|$). Second, given a mapping
  $d \colon V \to \{1,\dots,|U|\}$, we can determine the existence of
  hyperedges $e_i \in E$ with both $|d[e_i]| = |e_i|$ and $d[e_i] =
  c'[U_i]$ in constant depth.  Third, if $d \colon V \to
  \{1,\dots,|U|\}$ has Properties 1 and~2, so does any other $d'$ as
  long as it is identical to~$d$ on the vertices of $\bigcup_{i=1}^m
  e_i$. Since the number of vertices in $\bigcup_{i=1}^m e_i$ is at
  most $|U|$, which depends 
  only on the parameter~$|G|$, instantiating Fact~\ref{fact-cc} with
  $\{x_1,\dots,x_k\} = \bigcup_{i=1}^m e_i$ tells us that we can find
  one such $d'$ in constant depth.

  It remains to argue that the following two statements are equivalent:
  \begin{itemize}
  \item 
    There is a mapping $d\colon V \to \{1,\dots,|U|\}$ with the Properties
    1 and~2. 
  \item
    There is a mapping $c\colon U\to V$  such
    that is a proper coloring of $G$ and $c[U_i] \in E$ holds for all
    $i \in \{1,\dots,m\}$. 
  \end{itemize}

  For the first direction, let $d\colon V \to
  \{1,\dots,|U|\}$ be a mapping and $c' \colon U \to \{1,\dots,|U|\}$
  a proper coloring of~$G$ such that for each $i\in \{1,\dots,m\}$
  there is a hyperedge $e_i \in E$  with $|d[e_i]| = |e_i|$ and  $d[e_i] =
  c'[U_i]$. Define $c \colon U \to V$ as follows:
  We know that each $u \in U$ lies in a unique set~$U_i$, and since
  $d[e_i] = c'[U_i]$ for the hyperedge $e_i\in E$,
  there must be an element $v_u \in e_i$ with $d(v_u) =
  c'(u)$. Since $|d[e_i]| = |e_i|$, the element $v_u \in e_i$
  must be unique with this property. Define $c(u) =
  v_u$ and observe that $d \circ c = c'$ holds. (As an example,
  consider the functions $c'$ and~$d$ from Example~\ref{example-y5}
  and $u_5 \in U$ and $u_8 \in U$. For $u_5$, we have $u_5 \in U_2 =
  \{u_4,u_5,u_6\}$ and $d[e_2] = c'[U_2] = \{4,5,6\}$. The element
  $v_5$ has the property $d(v_5) = 5$ and, indeed, it is the only
  element with this property. Thus $c(u_5) = v_5$. For $u_8$, we have
  $u_8 \in U_3 = \{u_7,u_8\}$ and $d[e_3] = c'[U_3] = \{7,5\}$. Again,
  the element $v_5$ has the property $d(v_5) = 5$ and $c(u_8) = v_5$.) 

  The just-defined mapping~$c$ is a proper coloring of~$G$ since
  for every edge $\{u,v\} \in F$ we have $d(c(u)) = c'(u) \neq c'(v) = 
  d(c(v))$, which implies $c(u) \neq c(v)$. To prove $c[U_i] \in E$
  for $i \in \{1,\dots,m\}$, fix some~$i$. Since $d \circ c = c'$, we
  also have $d[c[U_i]] = c'[U_i]$ and this equals $d[e_i]$ 
  by assumption. By construction of~$c$, we clearly also have $c[U_i]
  \subseteq e_i$. Finally, since $|d[e_i]| = |e_i|$, we know that $d$
  is injective on~$e_i$. Now, from $d[c[U_i]] = d[e_i]$ and
  $c[U_i] \subseteq e_i$ and $d$ being injective on~$e_i$, we conclude
  that $c[U_i] = e_i \in E$ must hold. 

  For the second direction, let $c \colon U \to V$ be a proper coloring of~$G$
  with $c[U_i] \in E$ for all $i\in\{1,\dots,m\}$. We need to
  construct a mapping $d \colon V \to \{1,\dots,|U|\}$ and a proper
  coloring $c' \colon U \to \{1,\dots,|U|\}$ with the property that
  for each $i\in \{1,\dots,m\}$ there is a hyperedge $e_i \in E$ with
  $|d[e_i]| = |e_i|$ and  $d[e_i] = c'[U_i]$.

  For the definition of $d$, let $b \colon c[U] \to
  \{1,\dots,|c[U]|\}$ be a bijection. Define $d$ as follows: 
  \begin{align*}
    d(v) = \begin{cases}
      b(v) & \text{for $v\in c[U]$ and}\\
      1 & \text{otherwise.}
    \end{cases}
  \end{align*}
  Let $c' = d \circ c$. First, $c'$ is a proper coloring of~$G$ since
  $c$ is a proper coloring and $d$ restricted to $c[U]$ is a
  bijection. Now consider an $i \in \{1,\dots,m\}$. Then $e_i = 
  c[U_i]$ is a hyperedge in $E$. We have $|d[e_i]| = |e_i|$ since
  $d$ restricted to $e_i \subseteq c[U]$ is a bijection. We have
  $d[e_i] = d[c[U_i]] = c'[U_i]$ by definition of~$c'$. 
\end{proof}

We are now ready to prove Theorem~\ref{thm-ac0-pseudo}:

\begin{proof}[Proof of Theorem~\ref{thm-ac0-pseudo}]
  Recall that in order to prove the claim, we must  solve the
  following problem in constant depth: Given $H = (V,E)$, $k$, $L$,
  and a subset $C \subseteq e \in E$, check whether there exists a
  $T_L^k$-pseudo-sunflower $S$ of $H$ whose pseudo-core is exactly~$C$.
  We must now show how the existence of the pseudo-sunflower can be
  checked using the technical problem $\PLang[G]{restricted-coloring}$.

  The input for the restricted coloring problem will consist of a
  special graph~$G$ that encodes the different disjointedness properties
  of pseudo-sunflowers using edges and will consist of the hypergraph $H' = \bigl(V, \{\,e
  - C \mid e \supseteq C, e \in E\,\}\bigr)$. In other words, we
  restrict $H$ to those edges that contain the alleged core~$C$ (other
  edges cannot be part of the sought pseudo-sunflower anyway) and we
  remove the core from the edge since they all contain it.

  Let us now define the graph $G = (U,F)$. The objective of this
  definition is, of course, that there is a $T_L^k$-pseudo-sunflower~$S$
  with core $C$ if, and only if, there a proper coloring $c \colon U
  \to V$ of~$G$ such that $c[U_i] \in E(H')$ for all $i \in
  \{1,\dots,m\}$.

  Pseudo-sunflowers are mappings from
  $\operatorname{leafs}(T_k^L) \times \{0,\dots,L\}$ to subsets of~$V$
  such that for each leaf~$l$ the union $S(l,0) \cup 
  \dots \cup S(l,L)$ is a hyperedge in~$E$. In our case, we must have $S(l,1) \cup
  \dots \cup S(l,L) \in E(H')$ since $S(l,0) = C$ and we removed the
  fixed core~$C$ already from the hyperedges of~$H'$. In~$G$, we will
  have one set $U_l$ for each leaf~$l$ of $T_L^k$: 
  The vertices that will be assigned to the elements of $U_l$ by the
  coloring~$c$ should then form exactly the hyperedge $S(l,1) \cup
  \dots \cup S(l,L)$.

  If we knew that each $S(l,i)$ had size exactly~$1$, we could set $U
  = \operatorname{leafs}(T_k^L) \times \{1,\dots,L\}$: For each
  leaf~$l$ the coloring $c$ would need to pick $L$ vertices which,
  together, make up the hyperedge $S(l) - C$ of~$H'$. To ensure that
  $S(l,i)$ and $S(l,j)$ are disjoint for $i\neq j$, we would make each
  $\{l\} \times \{1,\dots,L\}$ a clique in~$G$. However, the
  sets $S(l,i)$ can have different sizes. For this reason, we do not
  use a single vertex in~$G$ for each $S(l,i)$, but $d$ different
  vertices (actually, $d-L+1$ vertices would suffice): The different
  elements of $S(l,i)$ can be represented by different vertices -- and
  if $|S(l,i)| < d$, the coloring $c$ can map the superfluous vertices
  to any of the vertices of $S(l,i)$.

  We set $U = \operatorname{leafs}(T_k^L) \times \{1,\dots,L\}
  \times \{1,\dots,d\}$ and define the partition of $U$ by $U_l =
  \{l\} \times \{1,\dots,L\} \times \{1,\dots,d\}$ for each $l \in
  \operatorname{leafs}(T_L^k)$. It remains to explain how we put edges
  into~$G$ such that the colorings of $G$ induce
  pseudo-sunflowers. The following edges are present in $G$ to ensure
  the four properties from Definition~\ref{def-pseudo-core}:
  \begin{enumerate}
  \item Nothing needs to be done to ensure the first property ($S(l,0)
    = C$) since $H'$ only contains hyperedges that used to contain~$C$. 
  \item Nothing needs to be done to ensure the second property ($S(l)
    \in E$) since $c[U_l] \in E(H')$ will ensure that $S(l) - C \in
    E(H')$ holds and, thus, $S(l) \in E(H)$.
  \item To ensure the third property ($S(l,i) \cap S(l,j) = \emptyset$
    for $i \neq j$), for each $l \in \operatorname{leafs}(T_L^k)$ and
    every $i \neq j$ and all $x,y \in \{1,\dots,d\}$ let $\{(l,i,x),
    (l,j,y)\}$ be an element of~$F$, that is, let it be an edge
    of~$G$.
  \item To ensure the fourth property ($S(l,z) \cap S(m,z)
    =\emptyset$ must hold when $l^z$ and $m^z$ have the same parent),
    for each $l,m \in \operatorname{leafs}(T_L^k)$ 
    and the smallest number $z$ with $l^z \neq m^z$ and all $x,y \in
    \{1,\dots,d\}$ let $\{(l,z,x), (m,z,y)\}$ be an element of~$F$.
  \end{enumerate}

  With this definition, we claim that $C$ is a $k$-pseudo-core of
  level~$L$ of $H$ if, and only if, $(H',G)$ is a element of
  $\PLang[G]{restricted-coloring}$. If we can show this, we are done by
  Lemma~\ref{lemma-restricted-cc}.

  We need to prove two directions. First, let a
  $T_L^K$-pseudo-sunflower~$S$ of $H$ with core $C$ be given. We must
  argue that there is a proper coloring $c \colon U \to V$ of $G$ with
  $c[U_l] \in E(H')$ for all leafs of $T_L^k$. This coloring is the
  following: Consider all leafs $l$ and all numbers $i \in
  \{1,\dots,L\}$. For each pair, the set $S(l,i)$ consist of some
  vertices $v_1,\dots,v_p \in V$ for some $p = |S(l,i)| \in \{1,\dots,d\}$. We set 
  $c(l,i,x) = v_x$ for $x \in \{1,\dots,p\}$ and $c(l,i,x) = v_p$  (or
  any other element of $S(l,i)$, it does not matter) for
  $x \in \{p+1,\dots,d\}$.

  With this definition, we clearly have $c[U_l] = S(l,1) \cup \dots
  \cup S(l,L)$ and since $S(l) \in E$, the latter is an element of
  $E(H')$. Furthermore, $c$ is a proper coloring: For all edges
  $\{(l,i,x), (l,j,y)\} \in F$ we know that the colors $c(l,i,x)$ and
  $c(l,j,y)$ are different since $c(l,i,x) \in S(l,i)$ and $c(l,j,y)
  \in S(l,j)$ and $S(l,i) \cap S(l,j) = \emptyset$. Next, for the
  edges of the form $\{(l,z,x), (m,z,y)\} \in F$ we also have that
  $c(l,z,x)$ and $c(m,z,y)$ are different since $S(l,z)$ and $S(m,z)$
  are disjoint.

  For the other direction, let a coloring $c$ be given. Define a
  mapping $S$ from $\operatorname{leafs}(T_L^k) \times \{0,\dots,L\}$
  to subsets of $V$ as follows: For all $l \in
  \operatorname{leafs}(T_l^k)$ let $S(l,0) = C$ and for
  $i\in\{1,\dots,L\}$ let $S(l,i) =
  \{c(l,i,1),c(l,i,2),\dots,c(l,i,d)\}$.

  To see that $S$ has the properties of a pseudo-sunflower, consider
  the  four properties. The first property is clearly true by
  definition. The second follows from $c[U_i] \in E(H')$ and, hence $C 
  \cup c[U_i] \in E(H)$. The third item follows from the following
  fact: For any two vertices $v_x \in S(l,i)$ and $v_y \in S(l,j)$ for
  $i\neq j$, there is an edge between $(l,i,x)$ and $(l,j,y)$ in~$G$
  and, thus, $v_x \neq v_y$. This shows that $S(l,i) \cap S(l,j) =
  \emptyset$ must hold; and note that, clearly, $S(l,i) \neq
  \emptyset$ always holds. For the fourth item, we have
  $S(l,z) \cap S(m,z) = \emptyset$ since for all $v_x \in S(l,z)$ and
  $v_y \in S(m,z)$ there is an edge between $(l,z,x)$ and $(m,z,y)$
  in~$G$.
\end{proof}
Theorem~\ref{thm-ac0-pseudo} now implies Theorem~\ref{theorem-main} by
simple standard arguments:

\begin{proof}[Proof of Theorem~\ref{theorem-main}]
  The only difference between the above claim and the claim of
  Theorem~\ref{thm-ac0-pseudo} (apart
  from the exact formulation) is that Theorem~\ref{theorem-main}
  requires the $\Class{AC}^0$-circuit family to have size
  $|V|^c|E|^c$ for some constant~$c$, while Theorem~\ref{thm-ac0-pseudo} allows it
  to have size $f(k,d) |V|^c |E|^c$. To reduce the size, on input
  $(H,k)$, a kernelization algorithm for Theorem~\ref{theorem-main}
  first checks  whether we have $f(k,d) > |V|^c|E|^c$ and, if so, just
  outputs $(H,k)$; otherwise it runs the kernelization algorithm from
  Theorem~\ref{thm-ac0-pseudo}, which needs size $f(k,d) |V|^c |E|^c \le
  |V|^{2c} |E|^{2c}$. 
\end{proof}

\section{Conclusion}

The results of this paper can be summarized as
$\PLang[k,d]{hitting-set} \in \Para\Class{AC}^0$ or, equivalently,
that kernels for the hitting set problem parameterized by $k$ and $d$
can be computed by a single $\Class{AC}^0$-circuit family. This result
refutes a conjecture of Chen et al.~\cite{ChenFH2017}. The proof
introduced a new technique: Iterated applications of color
coding can sometimes be ``collapsed'' into a single
application. This collapsing is not always straightforward (as the
present paper showed) and additional technical machinery may be needed
to make it work. 

The proof of our main result would be \emph{much} simpler if the
number of $k$-cores of a hypergraph depended only on the parameters
$k$ and~$d$ (since, then, only one round would be needed in the
parallel algorithm). While we gave examples that refute this hope, it
may be possible to tweak the idea a bit: We can 
compute in constant parallel time  the set of all
\emph{inclusion-minimal $k$-cores} of a hypergraph. We believe that we
can prove  that the number of these inclusion-minimal $k$-cores
depends only on $k$ and~$d$ (unfortunately, we need rather involved
and technical combinatorics and the dependence on $k$ and~$d$ seems to
be ``quite bad''). Nevertheless, if this is the case, we get a different 
proof that $\PLang[k,d]{hitting-set}$ has an $\Class{AC}^0$-kernelization,
where the complexity of proving correctness is shifted away from the
algorithm (which gets much simpler) towards the underlying graph
theory and combinatorics.


\bibliography{main}

\end{document}